\documentclass[runningheads,orivec]{llncs}
\usepackage[colorlinks,urlcolor=blue,linkcolor=blue,urlcolor=blue,citecolor=magenta]{hyperref}
\usepackage[T1]{fontenc}
\usepackage[utf8]{inputenc}
\usepackage{algorithmicx}
\usepackage{algpseudocode}
\usepackage{algorithm}
\usepackage{amsmath}
\usepackage{amssymb}
\usepackage{mathtools}
\usepackage{physics}
\usepackage{dsfont}
\usepackage{accents}
\usepackage{xfrac}
\usepackage{mathrsfs}
\usepackage{enumerate}
\usepackage{MnSymbol}
\usepackage{xcolor}
\usepackage[ligature, inference]{semantic}
\usepackage{ebproof}
\usepackage{graphicx}
\usepackage{stmaryrd}
\usepackage{xspace}
\usepackage{multirow}
\usepackage{array}
\usepackage{subcaption}
\usepackage{import}
\usepackage{adjustbox}
\usepackage{rotating}
\usepackage{stackengine}
\usepackage{marvosym}
\usepackage{multicol}
\usepackage{thm-restate}

\newcommand{\ie}{\textit{i.e.}\xspace}

\newcommand{\No}{\textnormal{Anc}}

\newcommand{\ia}{\mathrm{i}}
\newcommand{\ja}{\mathrm{j}}

\newcommand{\xa}{\mathrm{x}}
\newcommand{\sa}{\mathrm{s}}
\newcommand{\bb}{\mathrm{b}}
\newcommand{\ea}{\mathrm{e}}
\newcommand{\da}{\mathrm{d}}
\newcommand{\LN}{\mathcal{L}(\mathbb{N})}
\newcommand{\N}{\mathbb{N}}
\newcommand{\B}{\mathbb{B}}
\newcommand{\Z}{\mathbb{Z}}
\newcommand{\q}{\mathrm{q}}
\newcommand{\qa}{\mathrm{p}}
\renewcommand{\op}{\mathrm{op}}
\newcommand{\skp}{\tb{ skip};}
\newcommand{\tb}[1]{\text{\textnormal{\textbf{#1}}}}
\newcommand{\el}[1]{\textnormal{\small[$#1$]}}
\newcommand{\proc}{\mathrm{proc}}

\newcommand{\rec}{\mathrm{rec}}
\newcommand{\rot}{\mathrm{rot}}
\newcommand{\inv}{\mathrm{inv}}
\newcommand{\decl}{\tb{decl }}
\newcommand{\call}{\tb{call }}

\newcommand{\kpsi}{\ket{\psi}}
\newcommand{\kpsip}{\ket{\psi'}}
\newcommand{\kpsipp}{\ket{\psi''}}
\newcommand{\kpsik}{\ket{\psi_k}}
\newcommand{\qcase}[3]{\tb{qcase } #1 \tb{ of }\{0\to #2, 1\to #3\}}
\newcommand{\qcaseo}{\tb{qcase }}
\newcommand{\cif}[3]{\tb{if } #1 \tb{ then } #2 \tb{ else }#3}

\newcommand{\ST}{\mathrm{S}}
\newcommand{\CST}{\mathrm{Cst}}
\newcommand{\D}{\mathrm{D}}
\newcommand{\PR}{\mathrm{P}}
\newcommand{\QFT}{\mathrm{QFT}}
\newcommand{\U}{\mathrm{U}}
\newcommand{\asg}{{\ \texttt{\textasteriskcentered =}\ }}
\newcommand{\nil}{\textrm{nil}}
\newcommand{\tswap}{\textrm{SWAP}}
\newcommand{\cg}[3]{#1(#2,#3)}
\newcommand{\tms}{s}
\newcommand{\stdv}{\kpsi,A,l}

\newcommand{\RCalls}{\textnormal{width}}
\newcommand{\fbqp}{\textsc{fbqp}\xspace}

\newcommand{\bqp}{\textsc{bqp}\xspace}

\newcommand{\qfp}{\ensuremath{\widehat{\square^{\mathrm{QP}}_1}}}
\newcommand{\foq}{\textsc{foq}\xspace}
\newcommand{\wf}{\textsc{wf}\xspace}
\newcommand{\pfoq}{\textsc{pfoq}\xspace}
\newcommand{\size}[1]{| #1 |}
\newcommand{\sem}[1]{\llbracket #1 \rrbracket}
\newcommand{\topbot}{\diamond}
\newcommand{\semto}[1]{\stackrel{#1}{\longrightarrow}}
\newcommand{\compile}{\textbf{compile}\xspace}
\newcommand{\comprec}{\textbf{compr}\xspace}
\newcommand{\optimize}{\textbf{optimize}\xspace}

\newcommand{\cs}{cs}
\newcommand{\level}{\textnormal{level}}
\newcommand{\lv}{\textnormal{lv}}

\def\orcidID#1{\smash{\href{http://orcid.org/#1}{\protect\raisebox{-1.25pt}{\protect\includegraphics{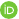}}}}}

\newcounter{program}

\newenvironment{program*}{\begin{equation*}\aligned}{\endaligned\end{equation*}}

\renewcommand\bra[1]{{\langle{#1}|}}
\makeatletter
\renewcommand\ket[1]{%
  \@ifnextchar\braket{\k@t{#1}\!\!}{\k@t{#1}\!}%
}
\newcommand\k@t[1]{{|{#1}\rangle}}
\makeatother

\begin{document}
\title{A programming language characterizing quantum polynomial time}
\author{Emmanuel Hainry\orcidID{0000-0002-9750-0460}
\and Romain Péchoux\orcidID{0000-0003-0601-5425}
\and Mário Silva\orcidID{0000-0002-9886-8400}}
\institute{Université de Lorraine, CNRS, Inria,
LORIA, F-54000 Nancy,
France
\email{\{hainry,pechoux,mmachado\}@loria.fr}}
\authorrunning{E. Hainry, R. Péchoux, M. Silva}

\maketitle

\begin{abstract}
We introduce a first-order quantum programming language, named \foq, whose terminating programs are reversible.
We restrict \foq to a strict and tractable subset, named \pfoq, of terminating programs with bounded width, that provides a first programming language-based characterization of the quantum complexity class \fbqp.
We finally present a tractable semantics-preserving algorithm compiling a \pfoq program to a quantum circuit of size polynomial in the number of input qubits.
\end{abstract}

\section{Introduction}
\paragraph{Motivations.} Quantum computing is an emerging and promising computational model that has been in the scientific limelight for several decades.
This phenomenon is mainly due to the advantage of quantum computers over their classical competitors, based on the use of purely quantum properties such as superposition and entanglement. The most notable example being Shor's algorithm for finding the prime factors of an integer~\cite{Shor}, which is exponentially faster than the most efficient known classical factoring algorithm and which is expected to have implications in cryptography (RSA encryption, etc.).

Whether due to the fragility of quantum systems, namely the engineering problem of maintaining a large number of qubits in a coherent state, or by lack of reliable technological alternatives, quantum computing is typically described at a level close to hardware. Without any hope of being exhaustive, one can think to quantum circuits~\cite{Feynman1982,NielsenChuang}, to measurement-based quantum computers~\cite{B09,DK06} or to circuit description languages~\cite{R15}. This low-level machinery restricts drastically the abstraction and programming ease offered by these models and quantum programs currently suffer from the comparison with their classical competitors, which have many high-level tools and formalisms based on more than 50 years of scientific research, engineering development, and practical and industrial applications. 

In order to solve these issues, a major effort is made to realize the promise of a quantum computer, which requires the development of different layers of hardware and software, together referred to as the \emph{quantum stack}. Our paper is part of this line of research. We focus on the highest layers of the quantum stack: quantum programming languages and quantum algorithms. We seek to better understand what can be done efficiently on a quantum computer and we are particularly interested in the development of quantum programming languages where program complexity can be certified automatically by some static analysis technique.

\paragraph{Contribution.}
Towards this end, we take the notion of polynomial time computation as our main object of study. Our contributions are the following.
\begin{itemize}
\item We introduce a quantum programming language, named \foq, that includes first-order recursive procedures. The input of a \foq program consist in a sorted set of qubits, a list of pairwise distinct qubit indexes. A \foq program can apply to each of its qubits basic operators corresponding to unary unitary operators. The considered set of operators has been chosen in accordance with~\cite{Yamakami20} to form a universal set of gates.
\item After showing that terminating \foq programs are reversible (Theorem~\ref{thm:rev}), we restrict programs to a strict subset, named \pfoq, for \emph{polynomial time} \foq. The restrictions put on a \pfoq programs are tractable (\ie, can be decided in polynomial time, see Theorem~\ref{thm:trac}), ensure that programs terminate on any input (Lemma~\ref{lem:termination}), and prevent programs from having any exponential blow up (Lemma~\ref{lem:poly}). 
\item We show that the class of functions computed by \pfoq programs is \emph{sound} and \emph{complete} for the quantum complexity class \fbqp. \fbqp is the functional extension of \emph{bounded-error quantum polynomial time}, known as \bqp~\cite{BernsteinVazirani}, the class of decision problems solvable by a quantum computer in polynomial time with an error probability of at most $\frac{1}{3}$ for all instances. Hence the language $\pfoq$ is, to our knowledge, the first programming language characterizing quantum polynomial time functions. Soundness (Theorem~\ref{thm:soundness}) is proved by showing that any \pfoq program can be simulated by a quantum Turing machine running in polynomial time~\cite{BernsteinVazirani}. The completeness of our characterization (Theorem~\ref{thm:completeness}) is demonstrated by showing that \pfoq programs strictly encompass Yamakami's function algebra, known to be \fbqp-complete~\cite{Yamakami20}.
\item We also describe a polynomial-time deterministic algorithm $\compile$ (based on the subroutines described in Algorithms~\ref{alg:comprec} and~\ref{alg:optimize}), that takes in a $\pfoq$ program $\PR$ and an integer $n$ and outputs a quantum circuit of size polynomial in $n$ that simulates $\PR$ on an input size of $n$ qubits. The existence of such circuits is not surprising, as a direct consequence of Yao's characterization of the class $\bqp$ in terms of uniform families of circuits of polynomial size~\cite{Yao1993}. However, a constructive  generation based on Yao's algorithm is not satisfactory because of the use of quantum Turing machines which makes the circuits complex and not optimal (in size). We show that, in our setting, circuits can be effectively computed and that the $\compile$ algorithm is tractable (Theorem~\ref{thm:compilepoly}).
\end{itemize}

Our programming language $\foq$ and the restriction to $\pfoq$ are illustrated throughout the paper, using the Quantum Fourier Transform $\QFT$ as a leading algorithm (Example~\ref{ex:qft}).

\paragraph{Related work.} 
This paper belongs to a long standing line of works trying to specify, understand, and analyze the semantics of quantum programming languages, starting with the cornerstone work of Selinger~\cite{selinger_2004}. The motivations in restricting the considered programs to \pfoq were inspired by the works on \emph{implicit computational complexity}, that seek to characterize complexity classes by putting restrictions (type systems or others) on standard programming languages and paradigms~\cite{Bellantoni1992,DL11,P20}. These restrictions have to be implicit (\ie, not provided by the programmer) and tractable. Among all these works, we are aware of two results~\cite{Yamakami20} and~\cite{DALLAGO2010377} studying polynomial time computations on quantum programming languages, works from which our paper was greatly inspired. \cite{DALLAGO2010377} provides a characterization of \bqp based on a quantum lambda-calculus. Our work is an extension to \fbqp with a restriction to first-order procedures.
Last but not least, \cite{DALLAGO2010377} is based on Yao's simulation of quantum Turing machines~\cite{Yao1993} while we provide an explicit algorithm for generating circuits of polynomial size. Our work is also inspired by the function algebra of~\cite{Yamakami20}, that characterizes \fbqp: our completeness proof shows that any function in~\cite{Yamakami20} can be simulated by a \pfoq program (Theorem~\ref{thm:completeness}). However, we claim that  \foq  is a more general language for \fbqp in so far that it is much less constraining (in terms of expressive power) than the function algebra of~\cite{Yamakami20}: any function of~\cite{Yamakami20} can be, by design, transformed into a \pfoq program, whereas the converse is not true. We can take as example the quantum Fourier transform (QFT) which, as noted in~\cite{Yamakami20}, cannot be exactly computed by the function algebra without an additional initial quantum function. Furthermore, the \emph{multi-qubit recursion} construction described in~\cite{Yamakami20} is more restrictive than what we allow in \pfoq{}, since we may only call the same recursive function in each branch. 

\section{First-order quantum programming language}

\paragraph{Syntax and well-formedness.}
We consider a quantum programming language, called \foq for First-Order Quantum programming language, that includes basic data types such as \text{Integers}, \text{Booleans}, \text{Qubits}, \text{Operators}, and \text{Sorted Sets} of qubits, lists of finite length where all elements are different. A \foq program has the ability to call first-order (recursive) procedures taking a sorted set of qubits as a parameter. Its syntax is provided in Figure~\ref{fig:syntax}.

Let $\xa$ denote an integer variable and $\bar{\qa}$, $\bar{\q}$ denote sorted sets variables. The size of the sorted set stored in $\bar{\q}$ will be denoted by $\size{\bar{\q}}$. 
We can refer to the $i$-th qubit in $\bar{\q}$ as $\bar{\q}\el{i}$, with $1 \leq i \leq \size{\bar{\q}}$.
Hence, each non-empty sorted set variable $\bar{\q}$ can be viewed as a list $[\bar{\q}\el{1},\dots, \bar{\q}\el{\size{\bar{\q}}}]$. The empty sorted set, of size $0$, will be denoted by $\nil$ and $\bar{\q}\ominus \el{i}$ will denote the sorted set obtained by removing the qubit of index $i$ in $\bar{\q}$. For notational convenience, we extend this notation by $\bar{\q} \ominus \el{i_1,\dots,i_k}$, for the list obtained by removing the qubits of indexes $i_1,\dots,i_k$ in the sorted set $\bar{\q}$.

The language also includes some constructs $\U^f$ to represent (unary) unitary operators, for some total function $f \in \Z \to [0,2\pi)\cap \Tilde{\mathbb{R}}$. The function $f$ is required to be polynomial-time approximable: its output is restricted to $\Tilde{\mathbb{R}}$, the set of real numbers that can be approximated by a Turing machine for any precision $2^{-k}$ in time polynomial in $k$.

\begin{figure}
\hrulefill
$$ \begin{array}{lllll}
\text{(Integers)} & \ia \ &\triangleq \ &\  n \ |\ \xa\ |\ \ia+n \ |\ \ia-n\ |\ \size{\sa}, \quad \text{with }n\in \mathbb{N}\\
\text{(Booleans)} & \bb \ &\triangleq \ &\  \ia>\ia\ |\  \ia \geq \ia \ | \ \ia=\ia \ | \ \bb \wedge \bb \ | \ \bb \vee \bb \ | \ \neg \bb \ \\
\text{(Sorted Sets)} & \sa \ &\triangleq \  & \ \nil \ |\  \bar{\q}\ |\   \sa\ominus \el{\ia} \\
\text{(Qubits)} & \q \ &\triangleq \  &\ \sa\el{\ia} \\
\text{(Operators)} & \U^f(\ia) \ &\triangleq \ &\ \textrm{NOT}\ |\ \textrm{R}_Y^f(\ia)\ |\ \textrm{Ph}^f(\ia) , \quad \text{with }f \in \Z \to [0,2\pi)\cap \Tilde{\mathbb{R}} \\ 
\text{(Statements)} & \ST\ &\triangleq \ & \skp\ |\  \q \asg \U^f(\ia); \ | \ \ST\ \ST \ |\ \tb{if }\bb\tb{ then }\ST\tb{ else } \ST  \\
& & & \ |\ \qcase{\q}{\ST}{\ST} \ |\ \call \proc\el{\ia}(\sa); &\\
\text{(Procedure declarations)} & \D\ &\triangleq \ & \ \varepsilon \ |\  \decl \proc\el{\xa}(\bar{\qa})\{\ST\},\,\D \\
\text{(Programs)} & \PR(\bar{\q})\ &\triangleq \ &\,\D::\ST 
\end{array}$$
\caption{Syntax of \foq programs}
\label{fig:syntax}
\hrulefill
\end{figure}

A  \foq \emph{program} $\PR(\bar{\q})$ consists of a sequence of \emph{procedure declarations} $\D$ followed by a \textit{program statement} $S$, $\varepsilon$ denoting the empty sequence. In what follows, we will sometimes refer to program $\PR(\bar{\q})$ simply as $\PR$. Let $var(\ST)$ be the set of variables appearing in the statement $\ST$. Let $\size{\PR}$ be the size of program $\PR$, that is the total number of symbols in $\PR$.

A procedure declaration $\decl \proc\el{\xa}(\bar{\qa})\{\ST\}$ takes a sorted set parameter $\bar{\qa}$ and some optional integer parameter $\xa$ as inputs. $\ST$ is called the \emph{procedure statement}, $\proc$ is the \emph{procedure name} and belongs to a countable set \text{Procedures}. We will write $\ST^{\proc}$ to refer to $\ST$ and $\proc \in \PR$ holds if $\proc$ is declared in $\D$.

Statements include a no-op instruction, applications of a unitary operator to a qubit ($\q \asg \U^f(\ia);$), sequences,  (classical) conditionals, \emph{quantum cases}, and \emph{procedure calls} ($\call \proc\el{\ia}(\sa);$). A quantum case $\qcase{\q}{\ST_0}{\ST_1}$ provides a quantum control feature that will execute statements $\ST_0$ and $\ST_1$ in superposition.
For example, the $CNOT$ gate on qubits $\bar{\q}\el{i}$ and $\bar{\q}\el{j}$, for $i,j \in \N$, $i \neq j$, can be simulated by the following statement:
\begin{equation*}\textrm{CNOT}(\bar{\q}\el{i},\bar{\q}\el{j}) \triangleq \qcase{\bar{\q}\el{i}}{\skp}{\bar{\q}\el{j}\asg \textrm{NOT};}.
\end{equation*}

Throughout the paper, we restrict our study to \emph{well-formed} programs, that is, programs $\PR=\D::\ST$ satisfying the following properties:
$var(\ST) \subseteq \{\bar{\q}\}$;
$\forall  \proc\in\PR$, $var(\ST^{\proc})\subseteq \{\xa,\bar{\qa}\}$;
procedure names declared in $\D$ are pairwise distinct;
for each procedure call, the procedure name is declared in $\D$.

\paragraph{Semantics.}
Let $\mathcal{H}_{2^n}$ be the \emph{Hilbert space} $\mathbb{C}^{2^n}$ of $n$ qubits. 
We use Dirac notation to denote a quantum state $\ket{\psi} \in \mathcal{H}_{2^n}$. Each $\ket{\psi} \in \mathcal{H}_{2^n}$ can be written as a superposition of bitstrings of size $n$: $\ket{\psi} = \sum_{w\in\{0,1\}^n}\alpha_w \ket{w}$, with $\alpha_w \in \mathbb{C}$ and $\sum_w |\alpha_w|^2=1$.
The \emph{length} $\ell(\ket{\psi})$ of the state $\ket{\psi}$ is $n$.
Given two matrices $M,N$, we denote by $M^\dagger$ the transpose conjugate of $M$ and by $M \otimes N$ the tensor product of $M$ by $N$.
$\bra{\psi}$ is equal to $\ket{\psi}^\dagger$ and $\ketbra{\psi}{\phi}$ and $\braket{\psi}{\phi}$ are respectively the inner product and outer product of  $\ket{\psi}$ and $\ket{\phi}$. Let $I_{n}$ be the identity matrix in $\mathbb{C}^{n \times n}$.
Given $m \leq n$ and $i \in \{0,1\}$, define $\ket{i}_m \triangleq I_{2^{m-1}} \otimes \ket{i} \otimes  I_{2^{n-m}}$ and $\bra{i}_m \triangleq (\ket{i}_m)^\dagger$. 

A function $\sem{\U^f}\in\mathbb{\Z}\to\Tilde{\mathbb{C}}^{2 \times 2}$ is associated to each $\U^f$ as follows:

\noindent$
\sem{\textrm{NOT}}(n)\triangleq \begin{pmatrix}
0 & 1\\
1 & 0
\end{pmatrix}\!, \ 
\sem{\textrm{R}_Y^f}(n)\triangleq\begin{pmatrix}
\cos(f(n)) & -\sin(f(n))\\
\sin(f(n)) & \cos(f(n))
\end{pmatrix}\!, \
\sem{\textrm{Ph}^f}(n)\triangleq
\begin{pmatrix}
1 & 0\\
0 & e^{i f(n)}
\end{pmatrix}\!,
$

where $\Tilde{\mathbb{C}}$ is the set of complex numbers whose both real and imaginary parts are in $\Tilde{\mathbb{R}}$.
One can check easily that each matrix $M\triangleq\sem{\U^f}(n) \in \Tilde{\mathbb{C}}^{2 \times 2}$ is unitary, {\ie}, it satisfies $M^\dagger\,M = M\,M^\dagger=I_2$.

Let $\B$ to be the set of Boolean values $b \in \{\tb{false},\tb{true}\}$.
For a given set $X$, let $\mathcal{L}(X)$ be the set of lists of elements in $X$. Let $l=[x_1,\ldots,x_m]$, with $x_1,\ldots,x_m \in X$, denote a list of $m$-elements in $\mathcal{L}(X)$ and $[\,]$ be the empty list (when $m=0$). For $l,l' \in \mathcal{L}(X)$, $l \MVAt l'$ denotes the concatenation of $l$ and $l'$. $hd(l)$ and $tl(l)$ represent the tail and the head of $l$, respectively. 
Lists of integers will be used to represent \text{Sorted Sets}. They contain pointers to qubits (\ie, indexes) in the global memory.

We interpret each basic data type $\tau$ as follows:
$\sem{\mathrm{Integers}} \triangleq \Z$, $\sem{\mathrm{Booleans}}\triangleq\B$, $\sem{\mathrm{Sorted Sets}} \triangleq \mathcal{L}(\mathbb{N})$, $\sem{\mathrm{Qubits}} \triangleq \mathbb{N}$, and $\sem{\mathrm{Operators}} \triangleq \Tilde{\mathbb{C}}^{2\times 2}$.
Each basic operation $\op \in \{+, -, >, \geq, =, \wedge, \vee, \neg \}$ of arity $n$, with $1 \leq n \leq 2$, has a type signature $\tau_1 \times \ldots \times \tau_n \to \tau$ fixed by the program syntax.
For example, the operation $+$ has signature $\mathrm{Integers}\times\mathrm{Integers} \to \mathrm{Integers}$.
A total function $\sem{\op}\in \sem{\tau_1}\times \ldots \times \sem{\tau_n} \to \sem{\tau}$ is associated to each $\op$.

For each basic type $\tau$, the reduction $\Downarrow_{\sem{\tau}}$ is a map in $\tau \times \mathcal{L}(\mathbb{N}) \to \sem{\tau}$.
Intuitively, it maps an expression of type $\tau$ to its value in $\sem{\tau}$ for a given list $l$ of pointers in memory. 
These reductions are defined in Figure~\ref{table:semnatbool}, where $\ea$ and $\da$ denote either an integer expression $\ia$ or a boolean expression $\bb$.

\begin{figure}[!h]
\hrulefill

\[
\begin{prooftree}
\hypo{(\ea, l) \Downarrow_{\sem{\tau_1}} m}
\hypo{(\da, l) \Downarrow_{\sem{\tau_2}} n}
\infer2[(Op)]{(\ea\ \op\ \da, l) \Downarrow_{\sem{\op}(\sem{\tau_1},\sem{\tau_2})} \sem{\op}(m,n)}
\end{prooftree}
\quad
  \begin{prooftree}
    \hypo{(\ia,l) \Downarrow_{\Z} n}
    \infer1[(Unit)]{(\U^f(\ia), l) \Downarrow_{\mathbb{C}^{2\times2}} \sem{\U^f}(n)}
  \end{prooftree}
\]
\\
\[
\begin{prooftree}
\hypo{}
\infer1[(Cst)]{(n,l)\Downarrow_{\Z} n}
\end{prooftree}
\quad
\begin{prooftree}
\hypo{(\sa, l)\Downarrow_{\mathcal{L}(\mathbb{N})} [x_1, \ldots, x_m]}
\hypo{(\ia, l) \Downarrow_{\Z} k \in [1,m]}
\infer2[(Rm$_\in$)]{(\sa\ominus\el{\ia}, l) \Downarrow_{\mathcal{L}(\mathbb{N})} [x_1, \ldots, x_{k-1}, x_{k+1}, \ldots, x_m]}
\end{prooftree}
\]
\\
\[
\begin{prooftree}
\hypo{(\sa,l)\Downarrow_{\mathcal{L}(\mathbb{N})} [x_1,\ldots,x_n]}
\infer1[(Size)]{(\size{\sa},l)\Downarrow_{\Z} n}
\end{prooftree}
\quad
\begin{prooftree}
\hypo{(\sa, l)\Downarrow_{\mathcal{L}(\mathbb{N})} [x_1, \ldots, x_m]}
\hypo{(\ia, l) \Downarrow_{\Z} k \notin [1,m]}
\infer2[(Rm$_{\notin}$)]{(\sa\ominus\el{\ia}, l) \Downarrow_{\mathcal{L}(\mathbb{N})} [\,]}
\end{prooftree}
\]
\\
\[
\begin{prooftree}
\infer0[(Nil)]{(\nil, l)\Downarrow_{\mathcal{L}(\mathbb{N})} [\,]}
\end{prooftree}
\quad
\begin{prooftree}
\hypo{(\sa, l)\Downarrow_{\mathcal{L}(\mathbb{N})} [x_1, \ldots, x_m]}
\hypo{(\ia, l) \Downarrow_{\Z} k \in [1,m]}
\infer2[(Qu$_\in$)]{(\sa\el{\ia}, l) \Downarrow_{\mathbb{N}} x_k}
\end{prooftree}
\]
\\
\[
\begin{prooftree}
\infer0[(Var)]{(\bar{\q}, l) \Downarrow_{\mathcal{L}(\mathbb{N})} l}
\end{prooftree}
\quad
\begin{prooftree}
\hypo{(\sa, l)\Downarrow_{\mathcal{L}(\mathbb{N})} [x_1, \ldots, x_m]}
\hypo{(\ia, l) \Downarrow_{\Z} k \notin [1,m]}
\infer2[(Qu$_{\notin}$)]{(\sa\el{\ia}, l) \Downarrow_{\mathbb{N}} 0}
\end{prooftree}
\]
\caption{Semantics of expressions}
\label{table:semnatbool}
\hrulefill
\end{figure}

Note that in rule (Rm$_{\notin}$), if we try to delete an undefined index then we return the empty list, and in rule (Qu$_{\notin}$), if we try to access an undefined qubit index then we return the value $0$ (defined indexes will always be positive).
The standard gates ${R}_Y(\pi/4)$, ${P}(\pi/4)$, and ${CNOT}$, form a universal set of gates \cite{Bokyn1999}, which justifies the choice of $\textrm{NOT}$, $\textrm{R}_Y^f(\ia)$, and $\textrm{Ph}^f(\ia)$ as basic operators. For instance, we can simulate the application of an Hadamard gate $H$ on $\q$ by the following statement $\q \asg \textrm{R}_Y^f(0);\ \q \asg \textrm{NOT};$, with the function $f$ defined by $\forall n, f(n)=\pi/4 \in [0,2\pi)\cap\Tilde{\mathbb{R}}$. By abuse of notation, we will sometimes use $\q \asg \textrm{H};$ to denote this statement. Using $\text{CNOT}$, we can also define the $\tswap$ operation swapping the state between two qubits $\bar{\q}\el{i}$ and $\bar{\q}\el{j}$, with $i,j \in \N,\ i \neq j$:
\begin{equation*}
\tswap(\bar{\q}\el{i},\bar{\q}\el{j})\triangleq \textrm{CNOT}(\bar{\q}\el{i},\bar{\q}\el{j})\ \textrm{CNOT}(\bar{\q}\el{j},\bar{\q}\el{i})\ \textrm{CNOT}(\bar{\q}\el{i},\bar{\q}\el{j}).\label{eq:swap}\end{equation*}

Let $\top$ and $\bot$ be two special symbols for termination and error, respectively, and let $\topbot$ stand for a symbol  in $\{\top,\bot\}$.
The set of \emph{configurations} of dimension $2^n$, denoted $\text{Conf}_n$, is defined by
\begin{equation*}
\text{Conf}_n \triangleq (\text{Statements}\cup \{\top,\bot\}) \times \mathcal{H}_{2^n} \times \mathcal{P}(\mathbb{N}) \times \mathcal{L}(\mathbb{N}),
\end{equation*}
with $\mathcal{P}(\mathbb{N})$ being the powerset over $\mathbb{N}$. A configuration $c=(\ST,\kpsi,A,l) \in \text{Conf}_n$ contains a statement $\ST$ to be executed (provided that $\ST \notin \{\top, \bot\}$), a quantum state $\kpsi$ of length $n$, a set $A$ containing the indexes of qubits that are allowed to be accessed by statement $\ST$, and a list $l$ of qubit pointers.

The program big-step semantics $\semto{}$, described in Figure~\ref{table:operationalsemantics}, is defined as a relation in $\bigcup_{n \in \mathbb{N}} \text{Conf}_n \times \text{Conf}_n$. 
In the rules of Figure~\ref{table:operationalsemantics}, $\semto{}$ is annotated by an integer, called \emph{level}.
For example, the level of the conclusion in the (Call$_{[\,]}$) rule is $1$.
The level is used to count the total number of procedure calls that are not in superposition (\ie, in distinct branches of a quantum case). 

\begin{figure}[!ht]
\hrulefill

\[
 \begin{prooftree}
 \infer0[(Skip)]{(\tb{skip},\stdv)\semto{0} (\top,\stdv)}
 \end{prooftree}
\]
\vspace{2mm}
\[
 \begin{prooftree}
 \hypo{(\sa\el{\ia},l)\Downarrow_{\mathbb{N}} n \notin A}
 \infer1[(Asg$_\bot$)]{(\sa\el{\ia}\asg \U^f(\ja);,\stdv)\semto{0} (\bot,\stdv)}
 \end{prooftree}
\]
\vspace{2mm}
\[
 \begin{prooftree}
 \hypo{(\sa\el{\ia},l)\Downarrow_{\mathbb{N}} n \in A}
 \hypo{(\U^f(\ja), l)\Downarrow_{\mathbb{C}^{2\times 2}} M}
 \infer2[(Asg$_\top$)]{(\sa\el{\ia}\asg \U^f(\ja);,\stdv)\semto{0} (  \top, I_{2^{n-1}}\otimes M \otimes I_{2^{l({\kpsi})-n}} \kpsi ,A, l )}
 \end{prooftree}
\]
\vspace{2mm}
\[
  \begin{prooftree}
  \hypo{(\ST_1,\stdv)\semto{m_1} (\top,\kpsip,A,l)}
   \hypo{(\ST_2,\kpsip,A,l)\semto{m_2} (\topbot,\kpsipp,A,l)}
 \infer2[(Seq$_\topbot$)]{(\ST_1\ \ST_2,\stdv)\semto{m_1+m_2} ( \topbot,\kpsipp,A,l)}
 \end{prooftree}
\]
\vspace{2mm}
\[
  \begin{prooftree}
  \hypo{( \ST_1,\stdv)\semto{m} (\bot,\stdv)}
  \infer1[(Seq$_\bot$)]{( \ST_1\ \ST_2,\stdv)\semto{m}(\bot,\stdv)}
 \end{prooftree}
\]
\vspace{2mm}
\[
  \begin{prooftree}
  \hypo{(\bb,l) \Downarrow_{\B} b \in \B}
  \hypo{(\ST_b,\stdv)\semto{m_b} (\topbot,\kpsip,A,l)}
  \infer2[(If)]{(\cif{\bb}{\ST_{\tb{true}}}{\ST_{\tb{false}}} ,\stdv)\semto{m_b}(\topbot ,\kpsip,A,l)}
  \end{prooftree}
\]
\vspace{2mm}
\[
  \begin{prooftree}
  \hypo{(\sa\el{\ia},l) \Downarrow_{\mathbb{N}} n \in A}
  \hypo{(\ST_{k},\kpsi,A\backslash \{{n}\},l)\semto{m_k} (\top,\kpsik,A\backslash \{{n}\},l)}
  \infer2[(Case$_\top$)]{ ( \qcase{\sa\el{\ia}}{\ST_0}{\ST_1},\stdv)\semto{\max_k m_k}(\top,\sum_{k} \ket{k}_{n}\! \bra{k}_{n} \kpsik,A,l)}
  \end{prooftree}
\]
\vspace{2mm}
\[
  \begin{prooftree}
  \hypo{(\sa\el{\ia},l) \Downarrow_{\mathbb{N}} n \in A}
  \hypo{(\ST_{k},\kpsi,A\backslash \{{n}\},l)\semto{m_k} (\topbot_k,\kpsik,A\backslash \{{n}\},l)}
\hypo{\bot\in\{\topbot_0, \topbot_1\}}
  \infer3[(Case$_\bot$)]{( \qcase{\sa\el{\ia}}{\ST_0}{\ST_1},\stdv)\semto{\max_k m_k}(\bot,\stdv)}
  \end{prooftree}
\]
\vspace{2mm}
\[
  \begin{prooftree}
  \hypo{(\sa\el{\ia},l) \Downarrow_{\mathbb{N}} n \notin A}
  \infer1[(Case$_{\notin}$)]{( \qcase{\sa\el{\ia}}{\ST_0}{\ST_1},\stdv)\semto{0}(\bot,\stdv)}
  \end{prooftree}
\]
\vspace{2mm}
\[
\begin{prooftree}
  \hypo{(\sa,l) \Downarrow_{\mathcal{L}(\mathbb{N})} l' \neq [\,]}
  \hypo{(\ia,l) \Downarrow_{\Z} n}
  \hypo{(\ST^{\proc}\{n/\xa\},\kpsi,A, l' )\semto{m} (\topbot,\kpsip,A,l')}
  \infer3[(Call$_{\topbot}$)]{(\call \proc\el{\ia}(\sa);,\stdv)\semto{m+1} (\topbot,\kpsip,A,l)}
\end{prooftree}
\]
\vspace{2mm}
\[
\begin{prooftree}
 \hypo{(\sa,l) \Downarrow_{\mathcal{L}(\mathbb{N})} [\,]}
 \infer1[(Call$_{[\,]}$) ]{(\call \proc\el{\ia}(\sa);,\stdv)\semto{1} (\top,\stdv)}
\end{prooftree}
\]
\caption{Semantics of statements}
\label{table:operationalsemantics}
\hrulefill
\end{figure}

We now give a brief intuition on the rules of Figure~\ref{table:operationalsemantics}.
Rules (Asg$_\bot$) and (Asg$_\top$) evaluate the application of a unitary operator, corresponding to $\U^f(\ja)$, to a qubit $\sa[\ia]$. For that purpose, they evaluate the index $n$ of $\sa[\ia]$ in the global memory. Rule (Asg$_\bot$) deals with the error case, where the corresponding qubit is not allowed to be accessed. Rule (Asg$_\top$) deals with the success case: the new quantum state is obtained by applying the result of tensoring the evaluation of $\U^f(\ja)$ to the right index.
Rules (Seq$_\topbot$) and (Seq$_\bot$) evaluate the sequence of statements, depending on whether an error occurs or not. The (If) rule deals with classical conditionals in a standard way. The three rules (Case$_\top$), (Case$_\bot$), and (Case$_{\notin}$) evaluate the qubit index $n$ of the control qubit $\sa\el{\ia}$. Then they check whether this index belongs to the set of accessible qubits (is $n$ in $A$?). If so, the two statements $\ST_0$ and $\ST_1$ are intuitively evaluated in superposition, on the projected state $\bra{0}_{n}\kpsi$ and $\bra{1}_{n}\kpsi$, respectively. During these evaluations, the index $n$ cannot be accessed anymore.
The rule (Call$_{[\,]}$) treats the base case of a procedure call when the sorted set parameter is empty. 
In the non-empty case, rule (Call$_\topbot$) evaluates the sorted set parameter $\sa$ to $l'$ and the integer parameter $\xa$ to $n$. It returns the result of evaluating the procedure statement $\ST^{\proc}\{n/\xa\}$, where $n$ has been substituted to $\xa$, w.r.t. the updated qubit pointers list $l'$.

For a given program $\PR=\D::\ST$ and a given quantum state $\kpsi \in \mathcal{H}_{2^n}$, the \textit{initial configuration} for input $\kpsi$ is $c_{init}(\kpsi) \triangleq (\ST,\kpsi,\{1,\dots,n\},[1,\dots,n]) \in \text{Conf}_{n}$. 
A program is \emph{error-free} if there is no initial configuration $c_{init}(\kpsi)$ such that $c_{init}(\kpsi) \semto{} (\bot,\kpsip,A,l)$.
We write $\sem{\PR}(\kpsi) = \kpsip$, whenever $c_{init}(\kpsi) \semto{m} (\top,\kpsip,A,l)$ holds for some $m$. $(\top,\kpsip,A,l)$ is called a \emph{terminal configuration}. Let $\mathcal{H}=\bigcup_n\mathcal{H}_{2^n}$, a program \emph{terminates} if $\sem{\PR}$ is a total function in $\mathcal{H} \to \mathcal{H}$. 
Note that if a program terminates then it is obviously error-free but the converse property does not hold. Every program $\PR$ can be efficiently transformed into an error-free program $\PR_{\neg\bot}$ such that $\forall \kpsi$, if $\sem{\PR}(\kpsi)$ is defined then $\sem{\PR}(\kpsi)=\sem{\PR_{\neg\bot}}(\kpsi)$. For example, an assignment $\sa\el{\ia} \asg \U^f(\ja);$ can be transformed into the conditional statement $\tb{if } ((0 < \ia) \wedge (\ia \leq \size{\sa}))\tb{ then }\sa\el{\ia} \asg \U^f(\ja);\tb{ else } \tb{skip};$.

\begin{example}\label{ex:qft} 
A notable example of quantum algorithm is the Quantum Fourier Transform (QFT), used as a subroutine in Shor's algorithm~\cite{Shor}, and whose quantum circuit is provided below, with $R_n \triangleq \sem{\textrm{Ph}^{\lambda x . \pi/2^{x-1}}}(n)$, for $n \geq 2$. After applying Hadamard and controlled $R_n$ gates, the circuit performs a permutation of qubits using swap gates.

\begin{center}
\includegraphics[width=\textwidth]{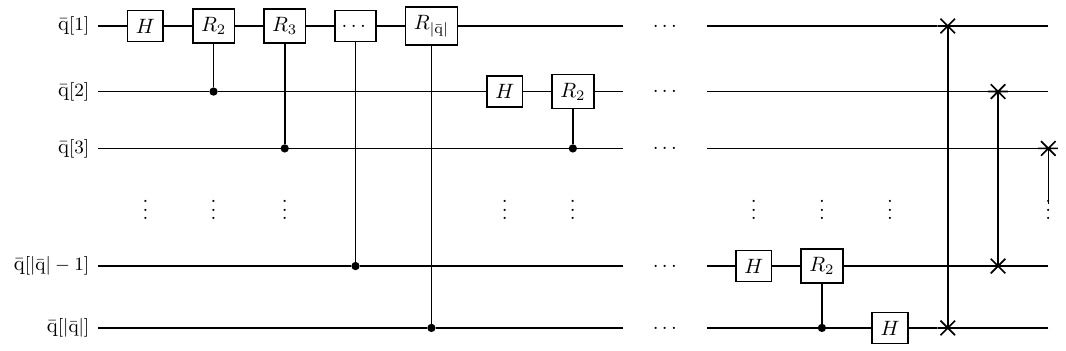}
\end{center}

Note that $\lambda x . \pi/2^{x-1}$ is a total function in $\Z\to[0, 2\pi) \cap \Tilde{\mathbb{R}}$. Hence, it is polynomial time approximable.
The above circuit can be simulated for any number of qubits $\size{\q}$ by the following \foq program $\QFT$.

\noindent\hrulefill \\
\noindent\begin{minipage}[t]{0.28\textwidth}
\begin{flushleft}
$\decl  \rec(\bar{\qa})\{$\\
$\quad   \bar{\qa}\el{1}\asg \textrm{H};$\\
$\quad  \call \rot\el{2}(\bar{\qa});$\\
$\quad   \call \rec(\bar{\qa}\ominus\el{1});\},$\\
$ $
\end{flushleft}
\end{minipage}
\begin{minipage}[t]{0.4\textwidth}
\begin{flushleft}
$\decl \rot\el{\xa}(\bar{\qa})\{$\\
$\ \ \tb{if }\size{\bar{\qa}}>1 \tb{ then}$\\
$\ \ \ \ \qcaseo \ \bar{\qa}\el{2} \tb{ of } \{$\\
$\ \ \ \ \ \ 0 \to \tb{skip}; $\\
$\ \ \ \ \ \ 1 \to \bar{\qa}\el{1}\asg \textrm{Ph}^{\lambda x. \pi/2^{x-1}}(\xa);$\\
$\ \ \ \ \}$\\
$\ \ \ \ \call \rot\el{\xa+1}(\bar{\qa}\ominus \el{2});$\\
$\ \ \tb{else skip};\},$\\
$ $
\end{flushleft}
\end{minipage}
\begin{minipage}[t]{0.3\textwidth}
\begin{flushleft}
$\decl \inv(\bar{\qa})\{$\\
$\quad \tb{if }\size{\bar{\qa}}>1 \tb{ then }$\\
$\quad \quad \tswap(\bar{\qa}\el{1},\bar{\qa}\el{\size{\bar{\qa}}});$\\
$\quad \quad \call \inv(\bar{\qa}\ominus\el{1,\size{\bar{\qa}}});$\\
$\quad \tb{else skip};\} ::$\\
\end{flushleft}
\end{minipage}
$\call \rec(\bar{\q});$ $\call \inv(\bar{\q});$

\noindent\hrulefill
\end{example}

\paragraph{Derivation tree and level.}
Given a configuration $c$ wrt a fixed program $\PR$, $\pi_\PR \rightslice c$ denotes the \emph{derivation tree} of $\PR$, the tree of root $c$ whose children are obtained by applying the rules of Figures~\ref{table:semnatbool} and \ref{table:operationalsemantics}  on configuration $c$ with respect to $\PR$.
We write $\pi$ instead of $\pi_\PR \rightslice c$ when $\PR$ and $c$ are clear from the context.
Note that a derivation tree $\pi$ can be infinite in the particular case of a non-terminating computation.
When $\pi'$ is finite, $\pi \unlhd \pi'$ denotes that $\pi$ is a subtree of $\pi'$. 

In the case of a terminating computation $\pi \rightslice c$, there exists a terminal configuration $c'$ and a level $m\in\N$ such that $c \semto{m} c'$ holds.
In this case, the level of $\pi$ is defined as $\lv_\pi \triangleq m$.
Given a \foq program $\PR$ that terminates, $\level_\PR$ is a total function in $\mathbb{N}\to\mathbb{N}$ defined as $\level_\PR(n) \triangleq \max_{\kpsi\in\mathcal{H}_{2^n}}\lv_{\pi_\PR \rightslice c_{init}(\kpsi)}$.

Intuitively, $\level_\PR(n)$ corresponds to the maximal number of non-superposed procedure calls in any program execution on an input of length $n$.

\begin{example}
Consider the program $\QFT$ of example~\ref{ex:qft}. Assume temporarily that $\QFT$ terminates (this will be shown in Example~\ref{ex:wf}).
For all $n \in \N$, $\level_{\QFT}(n)= \frac{(n+1)(n+2)}{2}+\lfloor{\frac{n}{2}} \rfloor +1$. Indeed, on sorted sets of size $n$, procedure $\rec$ is called recursively $n+1$ times and makes $n+1$ calls to procedure $\rot$ on sorted sets of size $n$, $n-1$, $\ldots$, and $1$. On sorted sets of size $n$, $\rot$ performs $n$ recursive calls. Hence the total number of calls to $\rot$ is equal to $\sum_{i=1}^{n} i$. Finally, on a sorted set of size $n$, procedure $\inv$ does $\lfloor{\frac{n}{2}} \rfloor +1$ recursive call.
\end{example}

A program $\PR$ is reversible if it terminates and there exists a program $\PR^{-1}$ such that $\sem{\PR^{-1}}\circ\sem{\PR} = Id$.

\begin{restatable}{thm}{thmrev}
\label{thm:rev}
All terminating \foq programs are reversible.
\end{restatable}

\section{Polynomial time soundness}\label{sec:soundness}

In this section, we restrict the set of \foq programs to a strict subset, named \pfoq, that is sound for the quantum complexity class \fbqp. For this, we define two criteria: a criterion ensuring that a program terminates and a criterion preventing a terminating program from having an exponential runtime.

\paragraph{Polynomial-time {\sc foq}.}
Given two statements $\ST,\ST'$, we write $\ST \in \ST'$ to mean that $\ST$ is a substatement of $\ST'$ and $\proc \in \ST$ holds if there are $\ia$ and $\sa$ such that $\call \proc\el{\ia}(\sa); \in \ST.$ Given a program $\PR=\D::\ST$, we define the relation $>_{\PR} \subseteq \textnormal{Procedures} \times \textnormal{Procedures}$ by $\proc_1 >_{\PR} \proc_2$ if $\proc_2 \in \ST^{\proc_1}$,  for any two procedures $\proc_1,$ $\proc_2 \in \ST$. Let the partial order $\succeq_{\PR}$ be the transitive and reflexive closure of $>_{\PR}$ and define the equivalence relation $\sim_{\PR}$ by $\proc_1 \sim_{\PR} \proc_2 $ if $\proc_1\succeq_{\PR} \proc_2$ and $\proc_2 \succeq_{\PR} \proc_1$ both hold. Define also the strict order $\succ_{\PR}$ by $\proc_1 \succ_{\PR} \proc_2 $ if $\proc_1\succeq_{\PR} \proc_2$ and  $\proc_1 \not\sim_{\PR} \proc_2$ both hold.

\begin{definition}\label{def:wf}
Let \wf  be the set of \foq programs $\PR$ that are error-free and satisfy the well-foundedness constraint: $\forall \proc \in \PR,\ \forall \call \proc'\el{\ia}(\sa); \in \ST^{\proc},$
$$\proc\sim_{\PR} \proc' \Rightarrow \exists k>0, \exists \ia_1,\dots,\ia_k,\ s  = \bar{\qa}\ominus\el{\ia_1,\dots,\ia_k}.
$$
\end{definition}

\begin{restatable}{lem}{lemtermination}
\label{lem:termination}
If $\PR \in \wf$, then $\PR$ terminates.
\end{restatable}

\begin{example}\label{ex:wf}
Consider the program $\QFT$ of Example~\ref{ex:qft}. The statements of the procedure declarations define the following relation: $\rec >_{\QFT} \rec$, $\rec >_{\QFT} \rot$, $\rot >_{\QFT} \rot$, and $\inv >_{\QFT} \inv$. Consequently, $\rec \sim_{\QFT} \rec $, $\rot \sim_{\QFT} \rot$, $\inv \sim_{\QFT} \inv$, and $\rec \succ_{\QFT} \rot $ hold. For each call to an equivalent procedure, we check that the argument decreases: $\bar{\qa}\ominus\el{1}$ in $\rec$, $\bar{\qa}\ominus \el{2}$ in $\rot$, and $\bar{\qa}\ominus\el{1,\size{\bar{\qa}}}$ in $\inv$. Consequently, $\QFT \in \wf$. We deduce from Theorem~\ref{lem:termination} that $\QFT$ terminates.
\end{example}

We now add a further restriction on mutually recursive procedure calls for guaranteeing polynomial time using a notion of width.
\begin{definition}
Given a program $\PR$ and a procedure $\proc \in \PR$, the \emph{width} of $\proc$ in $\PR$, noted $\RCalls_{\PR}(\proc)$, and the \emph{width} of $\proc$ in $\PR$ \emph{relatively to statement} $\ST$, noted $w_{\PR}^{\proc}(\ST)$, are two positive integers in  $\mathbb{N}$. They are defined inductively by:
\begin{align*}
&\RCalls_{\PR}(\proc) \triangleq w^{\proc}_{\PR}(\ST^{\proc}),\\ &
w_{\PR}^{\proc}(\tb{skip};)\triangleq 0,\\ &
w_{\PR}^{\proc}(\q\emph{\asg} \U^f(\ia);)\triangleq 0,\\ &
w_{\PR}^{\proc}(\ST_1 \ \ST_2)\triangleq w^{\proc}_{\PR}(\ST_1)+ w^{\proc}_{\PR}(\ST_2),\\ &
w_{\PR}^{\proc}(\cif{\bb}{\ST_{\tb{true}}}{\ST_{\tb{false}}})\triangleq \max(w^{\proc}_{\PR}(\ST_{\tb{true}}),w^{\proc}_{\PR}(\ST_{\tb{false}})),\\ &
w_{\PR}^{\proc}(\qcase{\q}{\ST_0}{\ST_1})\triangleq \max(w^{\proc}_{\PR}(\ST_0),w^{\proc}_{\PR}(\ST_1)),\\ &
w_{\PR}^{\proc}(\call \proc'\el{\ia}(\sa);)\triangleq \begin{cases} 1& \text{if }\proc\sim_{\PR} \proc',\\ 0 &\text{otherwise}. \end{cases}
\end{align*}
\end{definition}

\begin{definition}[PFOQ]\label{def:pfoq}
Let \pfoq be the set of programs $\PR$ in \wf that satisfy the following constraint:
$\forall \proc\in \PR, \RCalls_{\PR}(\proc)\leq 1$.
\end{definition}

\begin{example}
In the program of Example~\ref{ex:qft}, $\RCalls_{\QFT}(\rec) = \RCalls_{\QFT}(\rot) = \RCalls_{\QFT}(\inv) =1$, since $\rec \succ_{\QFT} \rot$ holds.
Since $\QFT \in \wf$, by Example~\ref{ex:wf}, we conclude that $\QFT$ is a \pfoq program.
\end{example}

We now show that the level of a \pfoq program is bounded by a	polynomial in the length of its input.

\begin{restatable}{lem}{lempoly}
\label{lem:poly}
For each $\pfoq$ program $\PR$, there exists a polynomial $Q \in \N[X]$ such that $\forall n \in \N,\ \level_\PR(n) \leq Q(n)$.
\end{restatable}

Moreover, checking whether a program is \pfoq is tractable.

\begin{restatable}{thm}{thmtrac}
\label{thm:trac}
For each \foq program $\PR$, it can be decided in time $O(\size{\PR}^2)$ whether $\PR_{\neg\bot} \in \pfoq$.
\end{restatable}

\paragraph{Quantum Turing machines and FBQP.}
Following Bernstein and Vazirani \cite{BernsteinVazirani}, a $k$-tape \emph{Quantum Turing Machine} (QTM), with $k \geq 1$, is defined by a triplet $(\Sigma,Q,\delta)$ where $\Sigma$ is a finite alphabet including a blank symbol $\#$, $Q$ is a finite set of states with an initial state $\tms_0$ and a final state $\tms_\top \neq \tms_0$, and $\delta$ is the quantum transition function in
$Q\times \Sigma^k \to \Tilde{\mathbb{C}}^{Q\times\Sigma^k\times\{L,N,R\}^k}$;
 $\{L,N,R\}$ being the set of possible movements of a head on a tape. Each tape of the QTM is two-way infinite and contains cells indexed by $\Z$.
A QTM successfully terminates if it reaches a superposition of only the final state $\tms_\top$. A QTM is said to be \textit{well-formed} if the transition function $\delta$ preserves the norm of the superposition (or, equivalently, if the time evolution of the machine is unitary).
The starting position of the tape heads is the \emph{start cell}, the cell indexed by $0$. If the machine terminates with all of its tape heads back on the start cells, it is called \textit{stationary}. We will use \textit{stationary} in the case where the machine terminates with its input tape head in the first cell, and all other tape heads in the last non-blank cell.  We will further refer to a QTM as being \textit{in normal form} if the only transitions from the final state $\tms_\top$ are towards the initial state $\tms_0$. These will be important conditions for the composition and branching constructions of QTMs. If a QTM is well-formed, stationary, and in normal form, we will call it \textit{conservative}~\cite{Yamakami20} (N.B.: our notion of stationary QTM differs but can be shown to be equivalent to the definition of stationary QTM in~\cite{Yamakami20}).

A configuration $\gamma$ of a $k$-tape QTM is a tuple $(s,\overline{w},\overline{n})$, where $s$ is a state in $Q$, $\overline{w}$ is a $k$-tuple of words in $\Sigma^*$, and $\overline{n}$ is a $k$-tuple of indexes (head positions) in $\Z$. An initial (final) configuration $\gamma_{init}$ (resp. $\gamma_{fin}$) is a configuration of the shape $(\tms_0,\overline{w},\overline{0})$ (resp. $(\tms_\top,\overline{w},\overline{0})$). We use $\gamma(w)$ to denote a configuration $\gamma$ where the word $w$ is written on the input/output tape. Following~\cite{BernsteinVazirani}, we write $\mathcal{S}$ to represent the inner-product space of finite complex linear combinations of configurations of the QTM $M$ with the Euclidean norm. 
A QTM $M$ defines a linear time operator $U_M : \mathcal{S} \to \mathcal{S}$, that outputs a superposition of configurations $\sum_i \alpha_i \ket{\gamma_i}$ obtained by applying a single-step transition of $M$ to a configuration $\ket{\gamma}$ (\ie, $U_M \ket{\gamma} = \sum_i \alpha_i \ket{\gamma_i}$). Let $U_M^t$, for $t \geq 1$, be the $t$-steps transition obtained from $U_M$ as follows: $U_M^1 \triangleq U_M$ and $U_M^{t+1} \triangleq U_M \circ U_M^t$.
Given a quantum state $\kpsi=\sum_{w \in \{0,1\}^{n}} \alpha_w\ket{w}$ and a configuration $\gamma$, let $\gamma(\kpsi) \in \mathcal{S}$ be the quantum configuration defined by $\gamma(\kpsi) \triangleq \sum_{w \in \{0,1\}^{n}} \alpha_w \ket{\gamma(w)}$.

A quantum function $f : \mathcal{H} \to \mathcal{H}$ is computed by the QTM $M$ in time $t$ if for any $\kpsi \in  \mathcal{H}$, 
$U_M^t (\gamma_{init}(\kpsi)) = \gamma_{fin}(f(\kpsip))$.
Given $T:\N\to\N$ and a quantum function $f$, we say that the QTM $M$ computes $f$ in time $T$ if for inputs of length $n$, $M$ computes $f$ in time $T(n)$.

\begin{definition}
Given two  functions $f:\{0, 1\}^\star\to\{0, 1\}^\star$, $F:\mathcal{H}\to\mathcal{H}$, and a value $p\in[0,1]$,
we say that $f$ is computed by $F$ with probability $p$ if $\forall x\in\{0, 1\}^\star,\ \size{\bra{f(x)}F(\ket{x})}^2 \geq p$.
\end{definition}
 
The class \fbqp is the functional extension of the complexity class \bqp. 

\begin{definition}[\cite{BernsteinVazirani}]\label{def:fbqp} A function $f \in \{0,1\}^\star \to \{0,1\}^\star$ is in \fbqp iff there exist a QTM $M$ and a polynomial $P \in \N[X]$ s.t.
$M$ computes $f$ in time $P$ with probability $\frac{2}{3}$.
\end{definition}

A function $f\in\{0,1\}^\star\to\{0,1\}^\star$ has a \emph{polynomial bound} $P\in\mathbb{N}[X]$  if $\forall n \in \N,  \forall x\in\{0,1\}^n, \exists k \leq P(n),\ f(x)\in\{0, 1\}^{k}$.
Functions in \fbqp have a polynomial bound as the size of their output is smaller than the polynomial time bound.

\paragraph{Soundness.}
We show that QTMs can simulate the function computed by any terminating \foq program. The time complexity of this simulation depends on the length of the input quantum state and on the level of the considered program.

\begin{restatable}{lem}{lemsimulation}
\label{lem:simulation}
For any terminating \foq program $\PR$, there exists a conservative QTM $M$ that computes $\sem{\PR}$ in time $O(n+n \times \level_\PR(n))$.
\end{restatable}

Now we show that any $\pfoq$ program computes a \fbqp function.

\begin{restatable}{thm}{thmsoundness}
\label{thm:soundness}
Given a \pfoq program $\PR$, a function $f:\{0, 1\}^\star\to\{0, 1\}^\star$, and a value $p\in(\frac{1}{2},1]$.
If $f$ is computed by $\sem{\PR}$ with probability $p$ then $f \in \fbqp$.
\end{restatable}
\begin{proof}
Using Lemma~\ref{lem:poly} and Lemma~\ref{lem:simulation}.\qed
\end{proof}

\section{{FBQP} completeness}\label{sec:includesFBQP}

In this section we show that any function in \fbqp can be faithfully approximated by a \pfoq program. Toward this end, we show that Yamakami's \cite{Yamakami20} \fbqp{}-complete function algebra can be exactly simulated in \pfoq.

\paragraph{Yamakami's function algebra.}\label{sec:Yamakami}
A characterization of \fbqp was provided in~\cite{Yamakami20} using a function algebra, named $\qfp$. 
Given a quantum state $\kpsi$ and a word $w \in \{0,1\}^n$, with $n \leq l(\kpsi)$. $\kpsi$ can be written as $\kpsi = \sum_i \alpha_i \ket{w_i z_i}$, with $w_i \in \{0,1\}^n$ and $z_i \in \{0,1\}^{l(\kpsi) -n}$. We write $\braket{w}{\psi}$ as an abuse of notation for the quantum state defined by $\braket{w}{\psi} \triangleq \sum_i \alpha_i \braket{w}{w_i} \ket{z_i}.$

\begin{definition}$\qfp$ is the smallest class of functions including the basic initial functions $\{I,Ph_\theta, Rot_\theta, NOT, SWAP \}$, with $\theta \in [0,2\pi)\cap \Tilde{\mathbb{C}}$, 
\begin{itemize}
\item $I(\kpsi)\triangleq\kpsi$
\item $Ph_\theta(\kpsi)\triangleq\ket{0}\braket{0}{\psi}+e^{\mathrm{i} \theta}\ket{1}\braket{1}{\psi}$
\item $Rot_\theta(\kpsi)\triangleq\cos\theta\kpsi+\sin\theta(\ket{1}\braket{0}{\psi}-\ket{0}\braket{1}{\psi})$
\item $NOT(\kpsi)\triangleq\ket{0}\braket{1}{\psi}+\ket{1}\braket{0}{\psi}$
\item $SWAP(\kpsi)\triangleq\begin{cases}\kpsi &\text{if }l(\kpsi)\leq 1\\
\sum_{a,b\in\{0,1\}}\ket{ba}\braket{ab}{\psi} & \text{otherwise}\end{cases}$
\end{itemize}
and closed under schemes $Comp$, $Branch$, and $kQRec_t$, for $k,t \in \N$,
\begin{itemize}
\item $Comp[F,G](\kpsi)\triangleq F(G(\kpsi))$
\item $Branch[F,G](\kpsi)\triangleq
\begin{cases}
\kpsi & \textnormal{if } l(\kpsi)\leq 1\\
\ket{0}\otimes F(\braket{0}{\psi})+\ket{1}\otimes G(\braket{1}{\psi}) &\textnormal{otherwise}
\end{cases}
$
\item  $kQRec_t[F,G,H](\kpsi)\triangleq
\begin{cases}
F(\kpsi) & \textnormal{if } l(\kpsi)\leq t\\
G\left(\sum_{w \in \{0,1\}^k}\ket{w}\otimes F_w(\bra{w}{H(\kpsi)})\right) &\textnormal{otherwise}
\end{cases}$

where each $F_w \in \{kQRec_t[F,G,H],I\}$.
\end{itemize}

\end{definition}

To handle general \fbqp functions, \cite{Yamakami20} defines the extended encoding of an input $x\in\{0,1\}^\star$ as $\phi_P(\ket{x})\triangleq \ket{0^{l(\ket{x})}1}\ket{0^{P(l({\ket{x}}))}10^{11P(l(\ket{x}))+6}1}\ket{x}$, for some polynomial $P \in \N[X]$ that is an upper bound on the output size of the desired \fbqp{} function.
$\phi_P$ simply consists in the quantum state $\ket{x}$ preceded by a polynomial number of ancilla qubits.
These ancilla provide space for internal computations and account for the polynomial bound associated to polynomial time QTMs.

\begin{restatable}[\cite{Yamakami20}]{thm}{labelyama} Given $f : \{0,1\}^\star \to \{0,1\}^\star$ with polynomial bound $P\in\N[X]$, the following statements are equivalent.
\begin{enumerate}[1.]
\item The function $f$ is in \fbqp.
\item There exists $F \in \qfp$ such that 
$F \circ \phi_P$ computes $f$ with probability $\frac{2}{3}$.
\end{enumerate}
\label{lemma:yamakami}
\end{restatable}

We show the following result by structural induction on a function in $\qfp$.

\begin{restatable}{thm}{thmQFPtoFOQ}
\label{thm:QFP->FOQ}
Let $F$ be a function in \qfp. Then there exists a \pfoq program $\PR$ such that $\sem{\PR} = F$.
\end{restatable}

We are now ready to state the completeness result.

\begin{restatable}{thm}{thmcompleteness}
\label{thm:completeness}
For every function $f$ in \fbqp with polynomial bound $Q \in \N[X]$, there is a \pfoq program $\PR$ such that $\sem{\PR}\circ \phi_Q$ computes $f$ with probability $\frac{2}{3}$.
\end{restatable}
\begin{proof}
By Theorem~\ref{lemma:yamakami} and Theorem~\ref{thm:QFP->FOQ}.\qed
\end{proof}

\section{Compilation to polynomial-size quantum circuits}
In this section, we provide an algorithm that compiles a \pfoq program on a given input length $n \in \N$ into a quantum circuit of size polynomial in $n$.

\emph{Quantum circuits}~\cite{D89} are a well-known graphical computational model for describing quantum computations. Qubits are represented by wires. Each unitary transformation $U$ acting on $n$ qubits can be represented as a gate $U$ with $n$ inputs and $n$ outputs. A circuit $C$ is an element of a PROP category (\cite{ML65}, a symmetric strict monoidal category) whose morphisms are generated by gates $G$ and wires. Let $\mathds{1}$ be the identity circuit (for any length) and $\circ$ and $\otimes$ be the composition and product, respectively. By abuse of notation, given $k$ circuits $C^1, \ldots, C^k$, $\circ_{i=1}^k C^i$ will denote the circuit $\tilde{C}^1 \circ \dots \circ \tilde{C}^k$, where each circuit $\tilde{C^i}$ is obtained by tensoring $C^i$ appropriately with identities so that the output of $C^i$ matches the input of $C^{i+1}$. By construction, a circuit is acyclic.
Each circuit $C_n$ can be indexed by its number $n \in \N$ of input wires (i.e., non ancilla qubits) and computes a function $\sem{C_n} \in \mathcal{H}_{2^n} \to \mathcal{H}_{2^n}$.
To deal with functions in $\mathcal{H}\to \mathcal{H}$, we consider families of circuits $(C_n)_{n\in\mathbb{N}}$, that are sequences of circuits such that each $C_n$ encodes computation on quantum states of length $n$. Hence each circuit has $n$ input qubits plus some extra ancilla qubits. These ancillas can be used to perform intermediate computations but also to represent functions whose output size is strictly greater than their input size. 
To avoid the consideration of families encoding undecidable properties, we put a uniformity restriction.

\begin{definition}
A family of circuits $(C_n)_{n\in\mathbb{N}}$ is said to be \emph{uniform} if there exists a polynomial time Turing machine that takes $n$ as input and outputs  a representation of $C_n$, for all $n\in\mathbb{N}$. 
\end{definition}

In quantifying the complexity of a circuit, it is necessary to specify the considered elementary gates, and define the complexity of an operation as the number of elementary gates needed to perform it. In our setting, we consider the following set of universal elementary gates $\{ {R}_Y(\pi/4), {P}(\pi/4),{CNOT}\}$. The size $\# C$ of a circuit $C$ is equal to the number of its gates and wires.

\begin{definition}
A family of circuits $(C_n)_{n\in\mathbb{N}}$ is said to be \emph{polynomial-size} with $\alpha\in \N \to \N$ ancilla qubits if there exists a polynomial $P \in \N[X]$ such that, for each $n\in\mathbb{N}$, $\# C_n \leq P(n)$ and the number of  ancilla qubits in $C_n$ is exactly $\alpha(n)$.
\end{definition}

Let $\chi_m:\mathcal{H}_{2^n}\to\mathcal{H}_{2^{n+m}}$ be defined by $\chi_m(\kpsi)\triangleq \kpsi\otimes\ket{0^m}$, for a state $\kpsi$ of size $n$.
Let $\xi_m: \mathcal{H}_{2^n}\to\mathcal{H}_{2^m}$, with $m\leq n$, be defined by $\xi_m(\ket{\psi}) \triangleq \sum_{w\in\{0, 1\}^{m}}\sum_{z\in\{0, 1\}^{n - m}}\braket{wz}{\psi}\ket{w}$.
Finally, let $\size{w}$, for $w \in \{0,1\}^\star$, be the size of the word $w$. 

\begin{restatable}{thm}{thm:yao}\tb{(Adapted from \cite{Yao1993} and \cite{NielsenChuang})}
\label{thm:yao}
A function $f: \{0, 1\}^{\star}\to\{0, 1\}^{\star}$ is in \fbqp iff there exists a uniform polynomial-size family of circuits $(C_n)_{n\in\mathbb{N}}$ with $\alpha$ ancilla qubits s.t.  $\forall x\in\{0,1\}^\star$, 
$\left|\braket{f(x)}{\xi_{\size{f(x)}}\circ\sem{C_{\size{x}}}\circ \chi_{\alpha(\size{x})}(\ket{x})}\right|^2\geq \frac{2}{3}$.

\end{restatable}

In Theorem~\ref{thm:yao}, $\sem{C_{\size{x}}}$ is a function in $\mathcal{H}_{2^{\size{x} + \alpha(\size{x})}} \to \mathcal{H}_{2^{\size{x} + \alpha(\size{x})}}$
The function $\chi_{\alpha(\size{x})}$ pads the input with ancilla in state $\ket{0}$ to match the circuit dimension.
The function $\xi_{\size{f(x)}}$ projects the output of the circuit to match the length of the function output $\size{f(x)}$.
Hence, for $\ket{x} \in \mathcal{H}_{2^{\size{x}}}$, $\xi_{\size{f(x)}}\circ\sem{C_{\size{x}}}\circ \chi_{\alpha(\size{x})}(\ket{x}) \in \mathcal{H}_{2^{\size{f(x)}}}$.

\paragraph{Compilation to circuits.}
For each \pfoq program $\PR$, the existence of a polynomial-size uniform family of circuits $(C_n)_{n\in\N}$ that computes $\sem{\PR}$ is entailed by the combination of Lemma~\ref{lem:poly} and Theorem~\ref{thm:yao}. However, due to the complex machinery of QTM, the constructions of both proofs cannot be used in practice to generate a circuit. In this section, we exhibit an algorithm that compiles directly a \pfoq program to a polynomial-size circuit. Note that this compilation process requires some care since recursive procedure calls in quantum cases may yield an exponential number of calls. 
The remainder of this section will be devoted to presenting an algorithm, named \compile, which, for a given \pfoq program $\PR$ and a given integer $n$ produces a circuit $C_n$ such that $\forall \kpsi \in \mathcal{H}_{2^{n}},\, \sem{\PR}(\kpsi) = \xi_{n}\circ\sem{C_n}\circ \chi_{\alpha(n)}(\kpsi)$. 

The \compile algorithm uses two subroutines, named \comprec and \optimize, and is defined by $\compile(\PR,n) \triangleq \comprec(\PR,[1,\ldots,n],\cdot)$. 

The subroutine \comprec (Algorithm~\ref{alg:comprec}) generates the circuit inductively on the program statement. It takes as inputs: a program $\PR$, a list of qubit pointers $l$, and a control structure $\cs$. A \emph{control structure} $cs$ is a partial function in $\N \to \{0,1\}$, mapping a qubit pointer to a control value (of a quantum case).  Let $\cdot$ be the control structure of empty domain. For $n \in \N$ and $k \in \{0,1\}$, $cs[n := k]$ is the control structure obtained from $cs$ by setting $cs(n) \triangleq k$. For a given $x\in\{0,1\}^\star$, we say that state $\ket{x}$ \emph{satisfies} $cs$ if, $\forall n\in dom(cs)$, $cs(n)=k\Rightarrow \size{\bra{k}_n\ket{x}}^2=1$. Two control structures $cs$ and $cs'$ are \emph{orthogonal} if there does not exist a state $\ket{x}$ that satisfies $cs$ and $cs'$. Note that if $\exists i \in dom(cs) \cap dom(cs'),\ cs(i)+cs'(i)=1$ then $cs$ and $cs'$ are orthogonal.

\begin{algorithm}
\caption{(\comprec)\\
\textbf{Input:} $(\PR,l,cs)\in \text{Programs}\times \LN \times (\mathbb{N} \to \{0,1\})$}\label{alg:comprec}
\begin{algorithmic}
\State{\tb{Let }$\D::\ST = \PR$\tb{ in}}
\If{$\ST= \skp$}
\State{$C\leftarrow \mathds{1}$}\Comment{Identity circuit}
\\
\ElsIf{$\ST= \sa\el{\ia}\asg\U^f(\ja);$\textbf{ and }$(\sa\el{\ia},l)\Downarrow_{\mathbb{N}}n$\textbf{ and }$(\U^f(\ja),l)\Downarrow_{\mathbb{C}^{2\times 2}}M$}
\State{$C\leftarrow M(cs,[n])$}\Comment{Controlled gate}
\\
\ElsIf{$\ST= \ST_1\,\ST_2$}
\State{$C\leftarrow \comprec(\D::\ST_1,l,cs)\circ \comprec(\D::\ST_2,l,cs)$}\Comment{Composition}
\\
\ElsIf{$\ST=\tb{if }\bb\tb{ then }\ST_{\tb{true}}\tb{ else }\ST_{\tb{false}} \textbf{ and }(\bb,l)\Downarrow_{\B}b$}
\State{$C\leftarrow \comprec(\D::\ST_b,l,cs)$}\Comment{Conditional}
\\
\ElsIf{$\ST=\qcase{\sa\el{\ia}}{\ST_0}{\ST_1}$\textbf{ and }$(\sa\el{\ia},l)\Downarrow_{\mathbb{N}}n$}
\State{$C \leftarrow \comprec(\D::\ST_0,l,cs[n:=0])\circ \,\comprec(\D::\ST_1,l,cs[n:=1])$}\Comment{Quantum case}
\\
\ElsIf{$\ST= \call \proc\el{\ia}(\sa)$\textbf{ and }$(\sa,l)\Downarrow_{\mathcal{L}(\mathbb{N})}[\,]$}
\State $ C\leftarrow\mathds{1}$\Comment{Nil call}
\\
\ElsIf{$\ST= \call \proc\el{\ia}(\sa)$\textbf{ and }$(\sa,l)\Downarrow_{\mathcal{L}(\mathbb{N})}l'\not=[\,]$\textbf{ and }$(\ia,l)\Downarrow_\Z n$}
\If{$\RCalls_{\PR}(\proc)=0$}
\State $ C\leftarrow\comprec(\D::\ST^{\proc}\{n/\xa\},l',cs)$\Comment{Non-recursive call}
\ElsIf{$\RCalls_{\PR}(\proc)=1$}
\State $ C\leftarrow\optimize(\D,[(cs,\ST^{\proc}\{n/\xa\})], \proc, l',\{\})$\Comment{Recursive call}
\EndIf
\EndIf\\

\Return{C}
\end{algorithmic}
\end{algorithm}

Given a control structure $cs$ and a statement $\ST$, a \emph{controlled statement} is a pair $(cs,\ST) \in \CST \triangleq (\mathbb{N}\to\{0,1\}) \times \textnormal{Statements}$. Intuitively, a controlled statement $(cs,\ST)$ denotes a statement controlled by the qubits whose indices are in $dom(cs)$. For a unitary gate $U\in\mathcal{H}_{2^n}\to\mathcal{H}_{2^n}$, a control structure $cs$, and a list of pointers $l=[x_1,\ldots,x_n] \in\LN$ such that $\{x_1,\ldots,x_n\}\cap dom(cs)=\emptyset$, $\cg{U}{cs}{l}$ denotes the circuit applying gate $U$ on qubits $\bar{\q}\el{x_1}, \ldots, \bar{\q}\el{x_n}$, whenever $\forall m \in dom(cs),$ $\bar{\q}\el{m}$ is in state $\ket{cs(m)}$. As demonstrated in~\cite{NielsenChuang}, this circuit can be built with $O(card(dom(cs)))$ elementary gates and ancillas, and a single controlled-$U$ gate. 

\begin{example}
As an illustrative example, consider a binary gate $U$ and a control structure $cs$ such that $dom(cs)=\{1,2,3\}$, $cs(1)=cs(2)=1$, and $cs(3)=0$. Also consider a list $l=[4,5] \in \LN$. The circuit $\cg{U}{cs}{l}$ is provided in Figure~\ref{fig:CU}.
\end{example}
\begin{figure}[t]
\begin{center}
\scalebox{0.7}{
\includegraphics{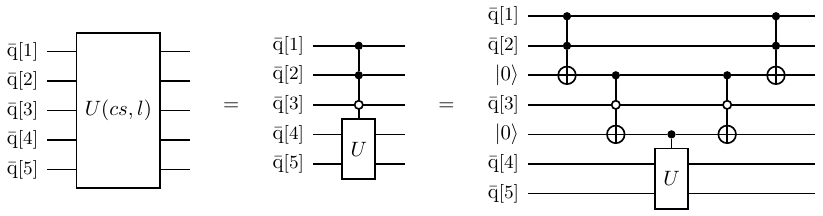}
}
\end{center}
\caption{Example of circuit $\cg{U}{cs}{l}$}\label{fig:CU}
\end{figure}
Similarly, we can define a generalized Toffoli gate as a circuit of the shape $NOT(cs,n)$. Since $card(dom(cs))$ will not scale with the size of the input, such a circuit has a constant cost in gates and ancillas and can thus be considered as an elementary gate.
We will also be interested in rearranging wires under a given control structure. For two lists of qubit pointers $l_1=[x_1,\ldots,x_n],\ l_2=[x'_1,\ldots,x'_n] \in \LN$, define ${SWAP}(cs,l_1,l_2)$ as the circuit that swaps the wires in $l_1$ with wires in $l_2$, controlled on $cs$. This circuit needs in the worst case one ancilla and $O(n)$ controlled $SWAP$ gates (also known as Fredkin gates).

 Let $\mathcal{D}\triangleq \mathcal{D}(\textnormal{Procedures}\times \Z \times \N\to \N \times \LN)$ be the set of dictionaries mapping keys of the shape $(\proc,i,j)$ to pairs of the shape $(a, l)$, where $i$ is the value of a classical parameter, $j$ is the size of a sorted set, and $a$ is a qubit index.
We will denote the empty dictionary by $\{\}$.
Let also $a \leftarrow {\tb{new }} ancilla()$ be an instruction that sets $a$ to a fresh qubit index.
 
The subroutine \optimize (Algorithm~\ref{alg:optimize}) treats the complex cases where circuit optimizations (merging) are needed, that is for recursive procedure calls.
It takes as input a sequence of procedure declarations $\D$, a list of controlled statements $l_{\CST}$, a procedure name $\proc$, a list of qubit pointers $l$, and a dictionary $\No$.
The subroutine iterates on list $l_{\CST}$ of controlled statements, indicating the statements left to be treated together with their control qubits. When recursive procedure calls appear in distinct branches of a quantum case, the algorithm merges these calls together. For that purpose, it uses new ancilla qubits as control qubits.
Given procedure calls of shape  $\call \proc[\ia](\sa);$, with respect to a given list $l \in \LN$, such that
$(\ia,l) \Downarrow_\Z i$, $(\sa, l)\Downarrow_{\LN} l'$, and $(\size{\sa},l) \Downarrow_\N j$.
If the key $(\proc,i,j)$ already exists in the dictionary $\No$, the associated ancilla is re-used, otherwise, 
$\No[\proc, i, j]$ is set to $(a,l')$. We can assume w.l.o.g. that the statement controlled on the ancilla can be treated only after all the re-uses of the ancilla. This can be done without increasing the total complexity of \optimize{}.

Some extra ancillas $e$ are also created for swapping wires and are not explicitly indexed since they are not revisited by the subroutine, and are just considered unique. 
Ancillas $a$ and $e$ are indexed and treated as input qubits, therefore they can be part of the domain of control structures.

\begin{algorithm}
\caption{\textbf{(\optimize)} Build circuit for recursive procedure $\proc$ \\
\textbf{Inputs:} $(\D,l_{\CST},\proc,l,\No)\in \textnormal{Decl}\times \mathcal{L}(\CST) \times \textnormal{Procedures} \times \LN \times \mathcal{D}$}\label{alg:optimize}

\begin{algorithmic}
\State $C_{\text{L}}\leftarrow\mathds{1}$; $C_{\text{R}}\leftarrow\mathds{1}$; $\PR\leftarrow \D::\tb{skip};$

\While{$l_{\CST}\neq [\,]$}
\State $(cs,\ST) \leftarrow hd(l_{\CST});\ l_{\CST} \leftarrow tl(l_{\CST})$\\

\If{$\ST=\ST_1\,\ST_2$}
\If{$w_\PR^\proc(\ST_1)=1$}
\State $l_{\CST}\leftarrow l_{\CST}\MVAt[(cs,\ST_1)];\ C_\text{R}\leftarrow \comprec(\D::\ST_2,l,cs)\circ C_\text{R}$
\Else
\State $l_{\CST}\leftarrow l_{\CST}\MVAt[(cs,\ST_2)];\ C_\text{L}\leftarrow C_\text{L}\circ \comprec(\D::\ST_1,l,cs)$
\EndIf
\EndIf\\

\If{$\ST=\tb{if }\bb\tb{ then }\ST_{\tb{true}}\tb{ else }\ST_{\tb{false}}\textbf{ and } (\bb,l) \Downarrow_\B b$}
\If{$w_\PR^\proc(\ST_b)=1$}
\State $l_{\CST}\leftarrow l_{\CST}\MVAt[(cs,\ST_b)]$
\Else
\State $C_L\leftarrow C_L\circ\comprec(\D::\ST_b,l,cs)$
\EndIf
\EndIf\\

\If{$\ST=\qcase{s\el{\ia}}{\ST_0}{\ST_1}$\textbf{ and }$(s\el{\ia},l)\Downarrow_{\mathbb{N}}n$}
\If{$w_\PR^\proc(\ST_0)=1 \textbf{ and } w_\PR^\proc(\ST_1)=1$}
\State $ l_{\CST}\leftarrow l_{\CST}\MVAt[(cs[n:=0],\ST_0),(cs[n:=1],\ST_1)]$
\ElsIf{$w_\PR^\proc(\ST_1)=0$}
\State $ l_{\CST}\leftarrow l_{\CST}\MVAt[(cs[n:=0],\ST_0)];$
\State $C_\text{R}\leftarrow\comprec(\D::\ST_1,l,cs[n:=1])\circ C_\text{R}$
\ElsIf{$w_\PR^\proc(\ST_0)=0$}
\State $l_{\CST}\leftarrow l_{\CST}\MVAt[(cs[n:=1],\ST_1)];$
\State $C_\text{R}\leftarrow \comprec(\D::\ST_0,l,cs[n:=0])\circ C_\text{R}$
\EndIf
\EndIf\\

\If{$\ST=\call \proc'\el{\ia}(\sa)$\tb{ and }$(\sa,l)\Downarrow_{\mathcal{L}(\mathbb{N})}l'\not=[\,]$\textbf{ and }$(\ia,l)\Downarrow_\Z n$}
\If{$(\proc',n,|l'|)\in\No$} 
\State{\tb{Let }$(a, l'') = \No[\proc',n,|l'|]$\tb{ in}}
\State $e \leftarrow \tb{new } ancilla();$
\State $C_\text{L}\leftarrow C_\text{L}\circ NOT(cs,e)\circ NOT(\cdot[e=1], a)\circ SWAP(\cdot[e=1],l',l'');$
\State $C_\text{R}\leftarrow SWAP(\cdot[e=1],l'',l')\circ NOT(\cdot[e=1], a)\circ NOT(cs,e) \circ C_\text{R}$ 
\Else 
\State $a \leftarrow \tb{new }ancilla()$
\State $\No[\proc',n,|l'|]\leftarrow (a,l');$
\State{$C_\text{L}\leftarrow C_\text{L}\circ NOT(cs,a);\ C_\text{R}\leftarrow NOT(cs,a)\circ C_\text{R};$}
\State{$l_{\CST}\leftarrow l_{\CST}\MVAt[(\cdot[a=1],\ST^{\proc'}\{n/\xa\})]$}

\EndIf
\EndIf
\EndWhile\\

\Return $C_L \circ C_R$
\end{algorithmic}
\end{algorithm}

\begin{restatable}{thm}{thm:circuits}
\label{thm:circuits}
For any $ \PR$ in \pfoq, there is $Q \in \N[X]$, $\forall n \in \N$, $\forall \kpsi \in \mathcal{H}_{2^{n}},$ $\sem{\PR}(\kpsi) = \xi_{n}\circ\sem{\tb{compile}(\PR,n)}\circ \chi_{\alpha(n)}(\kpsi)$ and $ \# \tb{compile}(\PR,n) \leq Q(n)$.
\end{restatable}

\begin{example}
$\compile(\QFT,n)$ outputs the circuit provided in Example~\ref{ex:qft}. Notice that there is no extra ancilla as no procedure call appears in the branch of a quantum case. 
\end{example}

\paragraph{Polynomial-size circuits.}
We show Theorem~\ref{thm:circuits} by exhibiting that any exponential growth of the circuit can be avoided by the \compile algorithm using an argument based on orthogonal control structures.
With a linear number of gates and a constant number of extra ancillas, we can merge calls referring to the same procedure, on different branches of a quantum case, when they are applied to sorted sets of equal size.
An example of the construction is given in Figure~\ref{fig:example-merging} where two instances of a gate $U$ are merged into one using $SWAP$ gates and gates controlled by orthogonal control structures.

\begin{figure}[h]
\begin{center}
\includegraphics[width=0.8\textwidth]{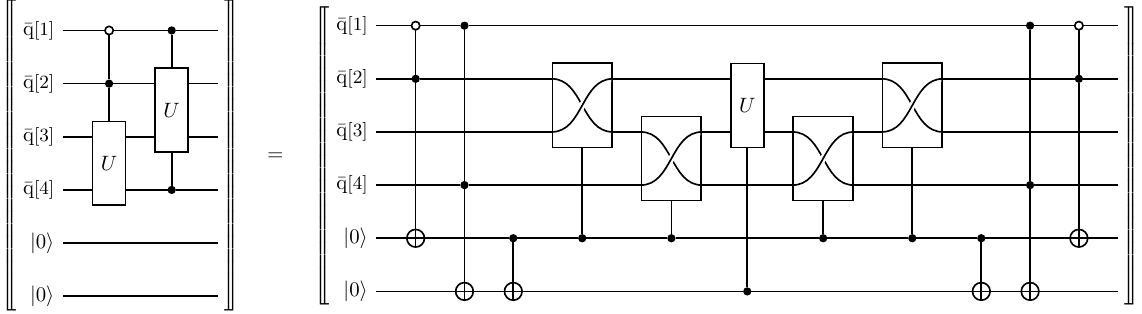}
\end{center}
\caption{Example of circuit optimization.}
\label{fig:example-merging}
\end{figure}

The following proposition shows that multiple uses of a gate can be merged in one provided they are applied to orthogonal control structures.
\begin{restatable}{lem}{joiningf}
\label{lem:joiningf}
For any circuit $C_n \triangleq \circ_{i=1}^k U(cs_i,l_i)$, with a unitary gate $U$, pairwise orthogonal $cs_1, \dots, cs_k\in\CST$, and $l_1, \dots l_k\in\LN$, there exists a circuit $C$ using one controlled gate $U$, $O(k n)$ gates, and $O(k)$ ancillas, and such that $\sem{C} = \sem{C_n}$.
\end{restatable}

Now we show that orthogonality  is an invariant property of \compile.

\begin{restatable}{lem}{lemortho}
\label{lem:ortho}
Orthogonality is an invariant property of the control structures in $l_{\CST}$ of the subroutine \tb{optimize}.
In other words, for any two distinct pairs $(cs, \ST)$, $(cs', \ST')$ in $l_{\CST}$, $cs$ and $cs'$ are orthogonal.
\label{thm:orthogonality}
\end{restatable}

\begin{restatable}{thm}{thmcompilepoly}
\label{thm:compilepoly}
For any $\PR$ in \pfoq{}, \tb{\compile}$(\PR,n)$ runs in time $O(n^{2\size{\PR}+1})$.
\end{restatable}

\begin{proof}
Using Lemma~\ref{lem:joiningf} and Lemma~\ref{lem:ortho}. 
\qed
\end{proof}

As there is no circuit duplication in the assignments of \compile, we can deduce from Theorem~\ref{thm:compilepoly} that the compiled circuit is of polynomial size.

\begin{corollary}\label{cor:cor}
For any $\PR$ in \pfoq{}, there exists a polynomial $Q\in\N[X]$ such that $\#{\tb{compile}}(\PR,n) \leq Q(n)$. 
\end{corollary}

\section{Acknowledgments}

This work is supported by the the Plan France 2030 through the PEPR integrated project
EPiQ ANR-22- PETQ-0007 and the HQI platform ANR-22-PNCQ-0002; and by the European
projects NEASQC and HPCQS.

\bibliographystyle{splncs04}

\bibliography{references}
\end{document}